\documentclass[USenglish]{article}
 
\usepackage{microtype}
\usepackage[utf8]{inputenc}
\usepackage{amssymb,amsmath,amsthm,mathtools}
\usepackage{sidecap}
\usepackage[sort&compress,sectionbib]{natbib}
\usepackage{tikz}
\usetikzlibrary{calc}
\usetikzlibrary{decorations.markings,decorations.pathmorphing}
\usetikzlibrary{shapes}
\usepackage{fullpage}
\usepackage{hyperref}
\usepackage{booktabs}
\usepackage{enumitem}
\usepackage{comment}
\usepackage{color}
\usepackage{boxedminipage}
\usepackage{caption,subcaption}
\usepackage{algorithm}
\usepackage{algpseudocode}
\usepackage{xspace}
\usepackage{wrapfig}
\usepackage{graphicx}
\usepackage[disable]{todonotes}
\usepackage{enumitem}

\DeclareMathOperator{\opt}{OPT}
\newcommand\prob[1]{\mathbf{Pr}\left[{#1}\right]}
\newcommand\expct[1]{\mathbf{E}\left[{#1}\right]}

\date{}
\title{New Algorithms for Maximum Disjoint Paths Based on Tree-Likeness}

\author{Krzysztof Fleszar\thanks{Universit{\"a}t W{\"u}rzburg, W{\"u}rzburg, Germany. \texttt{krzysztof.fleszar@uni-wuerzburg.de}}
        \and Matthias Mnich\thanks{Universit{\"a}t Bonn, Bonn, Germany. \texttt{mmnich@uni-bonn.de}. Supported by ERC Starting Grant 306465 (BeyondWorstCase).} \and Joachim Spoerhase\thanks{Universit{\"a}t W{\"u}rzburg, W{\"u}rzburg, Germany. \texttt{joachim.spoerhase@uni-wuerzburg.de}}
}

\newtheorem{theorem}{Theorem}
\newtheorem{lemma}{Lemma}
\newtheorem{definition}{Definition}
\newtheorem{proposition}{Proposition}
\newtheorem{conjecture}{Conjecture}

\newtheorem{observation}{Observation}

\newcommand{\poly}{\mathrm{poly}}

\newcommand{\blocked}{\textrm{\textit{blocked}}}
\newcommand{\free}{\textrm{\textit{free}}}
\newcommand{\toBeUsed}{\textrm{\textit{to-be-used}}}

\newcommand{\ch}[1]{c_{#1}}

\begin{document}

\maketitle

\begin{abstract}
  We study the classical $\mathsf{NP}$-hard problems of finding maximum-size subsets from given sets of~$k$ terminal pairs that can be routed via edge-disjoint paths ({\sc MaxEDP}) or node-disjoint paths ({\sc MaxNDP}) in a given graph.
  The approximability of {\sc MaxEDP/NDP} is currently not well understood; the best known lower bound is $\Omega(\log^{1/2 - \varepsilon}{n})$, assuming $\mathsf{NP}\not\subseteq\mathsf{ZPTIME}(n^{\poly \log n})$.
  This constitutes a significant gap to the best known approximation upper bound of~$\mathcal O(\sqrt{n})$ due to Chekuri et al. (2006) and closing this gap is currently one of the big open problems in approximation algorithms.  
  In their seminal paper, Raghavan and Thompson (Combinatorica, 1987) introduce the technique of randomized rounding for LPs; their technique gives an $\mathcal O(1)$-approximation when edges (or nodes) may be used by $\mathcal O\left(\frac{\log n}{\log\log n}\right)$ paths.

  \quad In this paper, we strengthen the above fundamental results.  We provide new bounds formulated in terms of the \emph{feedback vertex set number} $r$ of a graph, which measures its vertex deletion distance to a forest.  In particular, we obtain the following.
  \begin{itemize}
    \item For {\sc MaxEDP}, we give an $\mathcal{O}(\sqrt{r}\cdot \log^{1.5}{kr})$-approximation algorithm.
  As $r\leq n$, up to logarithmic factors, our result strengthens the best known ratio~$\mathcal O(\sqrt{n})$ due to Chekuri et al.    
\item Further, we show how to route $\Omega(\opt)$ pairs with congestion $\mathcal O\left(\frac{\log{kr}}{\log\log{kr}}\right)$, strengthening the bound obtained by the classic approach of Raghavan and Thompson.
   \item For {\sc MaxNDP}, we give an algorithm that gives the optimal answer in time~$(k+r)^{\mathcal O(r)}\cdot n$.
     If~$r$ is at most triple-exponential in~$k$, this improves the best known algorithm for {\sc MaxNDP} with parameter~$k$, by Kawarabayashi and Wollan (STOC 2010).
   \end{itemize}

  We complement these positive results by proving that {\sc MaxEDP} is $\mathsf{NP}$-hard even for $r=1$, and {\sc MaxNDP} is $\mathsf{W}[1]$-hard for parameter $r$. This shows that neither problem is fixed-parameter tractable in~$r$ unless $\mathsf{FPT} = \mathsf{W}[1]$ and that our approximability results are relevant even for very small constant values of $r$.
\end{abstract}

\section{Introduction}
\label{sec:introduction}
In this paper, we study disjoint paths routing problems.
In this setting, we are given an undirected graph $G$ and a collection of source-destination pairs $\mathcal M = \{(s_1, t_1), \hdots, (s_k, t_k)\}$.
The goal is to select a maximum-sized subset $\mathcal M' \subseteq \mathcal M$ of the pairs that can be \emph{routed}, where a routing of $\mathcal M'$ is a collection $\mathcal P$ of paths such that, for each pair $(s_i, t_i) \in \mathcal M'$, there is a path in $\mathcal P$ connecting $s_i$ to~$t_i$.
In the {\sc Maximum Edge Disjoint Paths (MaxEDP)} problem, a routing~$\mathcal P$ is \emph{feasible} if its paths are pairwise edge-disjoint, and in the {\sc Maximum Node Disjoint Paths (MaxNDP)} problem the paths in $\mathcal P$ must be pairwise vertex-disjoint.

Disjoint paths problems are fundamental problems with a long history and significant connections to optimization and structural graph theory.
The decision version of {\sc MaxEDP}/{\sc MaxNPD} asks whether all of the pairs can be routed.
Karp~\cite{Karp1975} showed that, when the number of pairs is part of the input, the decision problem is $\mathsf{NP}$-complete.
In undirected graphs, {\sc MaxEDP} and {\sc MaxNDP} are solvable in polynomial time when the number of pairs is a fixed constant; this is a very deep result of Robertson and Seymour~\cite{RobertsonSeymour1995} that builds on several fundamental results in structural graph theory from their graph minors project.

In this paper, we consider the optimization problems {\sc MaxEDP} and {\sc MaxNDP} when the number of pairs are part of the input. 
In this setting, the best approximation ratio for {\sc MaxEDP} is achieved by an~$\mathcal O(\sqrt{n})$-approximation~algorithm~\cite{ChekuriEtAl2006,KolliopoulosS04}, 
where $n$ is the number of nodes, whereas the best hardness for undirected graphs is only $\Omega(\log^{1/2 - \varepsilon}{n})$ \cite{andrews2010inapproximability}.
Bridging this gap is a fundamental open problem that seems quite challenging at the moment.

Most of the results for routing on disjoint paths use a natural multi-commodity flow relaxation as a starting point.
A well-known integrality gap instance due to Garg et al.~\cite{GargEtAl1997} shows that this relaxation has an integrality gap of $\Omega(\sqrt{n})$, and this is the main obstacle for improving the $\mathcal O(\sqrt{n})$-approximation~ratio in general graphs.
The integrality instance on an $n \times n$ grid (of treewidth $\Theta(\sqrt{n})$) exploits a topological obstruction in the plane that prevents a large integral routing; see Fig.~\ref{fig:lp}.
This led Chekuri et al. \cite{cns-tw} to studying the approximability of {\sc MaxEDP} with respect to the \emph{tree-width} of the underlying graph.
In particular, they pose the following conjecture:
\begin{conjecture}[\cite{chekuri2009note}]
\label{conj:tw}
  The integrality gap of the standard multi-commodity flow relaxation for {\sc MaxEDP} is~$\Theta(w)$, where $w$ is the treewidth of the graph.
\end{conjecture}
Recently, Ene et al.~\cite{EneEtAl2016} showed that {\sc MaxEDP} admits an $\mathcal O(w^3)$-approximation algorithm on graphs of treewidth at most $w$.
Theirs is the best known approximation ratio in terms of $w$, improving on an earlier~$\mathcal O(w\cdot 3^w)$-approximation algorithm due to Chekuri et al.
This shows that the problem seems more amenable on ``tree-like'' graphs.

However, for $w=\omega(n^{1/6})$, the bound is weaker than the bound of $\mathcal O(\sqrt{n})$.
In fact, {\sc EDP} remains $\mathsf{NP}$-hard even for graphs of \emph{constant} treewidth, namely treewidth $w = 2$~\cite{NishizekiEtAl2001}.
This further rules out the existence of a fixed-parameter algorithm for {\sc MaxEDP} parameterized by~$w$, assuming $\mathsf{P}\not=\mathsf{NP}$.
Therefore, to obtain fixed-parameter tractability results as well as better approximation guarantees, one needs to resort to parameters stronger than treewidth.

Another route to bridge the large gap between approximation lower and upper bounds for {\sc MaxEDP} is to allow the paths to have \emph{low congestion} $c$: that is, instead of requiring the routed paths to be pairwise disjoint, at most $c$ paths can use an edge.
In their groundbreaking work, Raghavan and Thompson~\cite{RaghavanThompson1987} introduced the technique of randomized rounding of LPs to obtain polynomial-time approximation algorithms for combinatorial problems. Their approach allows to route $\Omega(\opt)$ pairs of paths with congestion $\mathcal O\left(\frac{\log{n}}{\log\log{n}}\right)$.
This extensive line of research \cite{andrews:low-congestion,chuzhoy:edp-const-congestion,KawarabayashiKobayashi2011}  has culminated in a $\log^{\mathcal O(1)} k$-approximation algorithm with congestion~$2$ for {\sc MaxEDP}~\cite{chuzhoy2012polylogarithmic}.  A slightly weaker result also holds for \textsc{MaxNDP}~\cite{chekuri-ene:ndp-congestion}.

\subsection{Motivation and Contribution}
The goal of this work is to study disjoint paths problems under another natural measure for how ``far'' a graph is from being a tree. In particular, we propose to examine \textsc{MaxEDP} and \textsc{MaxNDP} under the \emph{feedback vertex set number}, which for a graph~$G$ denotes the smallest size $r$ of a set $R$ of $G$ for which $G - R$ is a forest.  Note that the treewidth of $G$ is at most $r+1$.  Therefore, given the $\mathsf{NP}$-hardness of {\sc EDP} for $w = 2$ and the current gap between the best known upper bound $\mathcal O(w^3)$ and the linear upper bound suggested by Conjecture~\ref{conj:tw}, it is interesting to study the stronger restriction of bounding the feedback vertex set number $r$ of the input graph.
Our approach is further motivated by the fact that \textsc{MaxEDP} is efficiently solvable on trees by means of the algorithm of Garg, Vazirani and Yannakakis~\cite{GargEtAl1997}. Similarly, \textsc{MaxNDP} is easy on trees (see Theorem~\ref{thm:maxndp-fpt-kr}).

Our main insight is that one can in fact obtain bounds in terms of $r$ that either strengthen the best known bounds or are almost tight (see Table~\ref{tab:results}). It therefore seems that parameter~$r$ correlates quite well with the ``difficulty'' of disjoint paths problems.

Our first result allows the paths to have small congestion: in this setting, we strengthen the result, obtained by the classic randomized LP-rounding approach of Raghavan and Thompson~\cite{RaghavanThompson1987}, that one can always route~$\Omega(\opt)$ pairs with congestion $\mathcal O\left(\frac{\log{n}}{\log\log{n}}\right)$ with constant probability.
\begin{theorem}
\label{thm:highprob}
For any instance $(G,\mathcal M)$ of {\sc MaxEDP}, one can efficiently find a routing of $\Omega(\opt)$ pairs with congestion $\mathcal O\left(\frac{\log{kr}}{\log\log{kr}}\right)$ with constant probability; in other words, there is an efficient $\mathcal O(1)$-approximation algorithm for {\sc MaxEDP} with congestion $\mathcal O\left(\frac{\log{kr}}{\log\log{kr}}\right)$.
\end{theorem}

Our second main result builds upon Theorem~\ref{thm:highprob} and uses it as a subroutine.
We show how to use a routing for {\sc MaxEDP} with low congestion to obtain a polynomial-time approximation algorithm for {\sc MaxEDP} \emph{without congestion} that performs well in terms of $r$.
\begin{theorem}
\label{thm:fvs}
  The integrality gap of the multi-commodity flow relaxation for {\sc MaxEDP} with~$k$ terminal pairs is  $\mathcal{O}(\sqrt{r}\cdot\log^{1.5} rk)$ for graphs with feedback vertex set number~$r$.
  Moreover, there is a polynomial time algorithm that, given a fractional solution to the relaxation of value $\mathsf{opt}$, it constructs an integral routing of size  $\mathsf{opt} / \mathcal{O}(\sqrt{r}\cdot\log^{1.5} rk)$. 
\end{theorem}
In particular, our algorithm strengthens the best known approximation algorithm for {\sc MaxEDP} on general graphs~\cite{ChekuriEtAl2006} as always $r\leq n$, and indeed it matches that algorithm's performance up to polylogarithmic factors.  Substantially improving upon our bounds would also improve the current state of the art of \textsc{MaxEDP}.  Conversely, the result implies that it suffices to study graphs with close to linear feedback vertex set number in order to improve the currently best upper bound of~$\mathcal{O}(\sqrt{n})$ on the approximation ratio \cite{ChekuriEtAl2006}. 

Our algorithmic approaches harness the forest structure of $G - R$ for any feedback vertex set $R$.
However, the technical challenge comes from the fact that the edge set running between $G - R$ and $R$ is unrestricted.
Therefore, the ``interaction'' between $R$ and $G - R$ is non-trivial, and flow paths may run between the two parts in an arbitrary manner and multiple times.  In fact, we show that \textsc{MaxEDP} is already $\mathsf{NP}$-hard if $R$ consists of a \emph{single node} (Theorem~\ref{thm:edp-nphard-r2}); this contrasts the efficient solvability on forests~\cite{GargEtAl1997}.

In order to overcome the technical hurdles we propose several new concepts, which we believe could be of interest in future studies of disjoint paths or routing problems.

In the randomized rounding approach of Raghavan and Thompson \cite{RaghavanThompson1987}, it is shown that the probability that the congestion on any fixed edge is larger than $c\frac{\log n}{\log\log n}$ for some constant~$c$ is at most $1/n^{O(1)}$. Combining this with the fact that there are at most $n^2$ edges, yields that every edge has bounded congestion w.h.p. The number of edges in the graph may, however, be unbounded in terms of $r$ and $k$. Hence, in order to to prove Theorem~\ref{thm:highprob}, we propose a non-trivial \emph{pre-processing step} of the optimum LP solution that is applied prior to the randomized rounding.  In this step, we aggregate the flow paths by a careful rerouting so that the flow ``concentrates'' in $O(kr^2)$ nodes (so-called \emph{hot spots}) in the sense that if all edges incident on hot spots have low congestion then so have all edges in the graph.  Unfortunately, for any such hot spot the number of incident edges carrying flow may still be unbounded in terms of $k$ and $r$.  We are, however, able to give a refined probabilistic analysis that suitably relates the probability that the congestion bound is exceeded to the amount of flow on that edge.  Since the total amount of flow on each hot spot is bounded in terms of $k$, the probability that \emph{all} edges incident on the same hot spot have bounded congestion is inverse polynomial in~$r$ and~$k$.

The known $\mathcal{O}(\sqrt{n})$-approximation algorithm for \textsc{MaxEDP} by Chekuri et al. \cite{ChekuriEtAl2006} employs a clever LP-rounding approach. If there are many long paths then there must be a single node carrying a significant fraction of the total flow and a good fraction of this flow can be realized by integral paths by solving a single-source flow problem. If the LP solution contains many short flow paths then greedily routing these short paths yields the bound since each such path blocks a bounded amount of flow.  In order to prove Theorem~\ref{thm:fvs}, it is natural to consider the case where there are many paths visiting a large number of nodes in $R$.  In this case, we reduce to a single-source flow problem, similarly to the approach of Chekuri et al.
The case where a majority of the flow paths visit only a few nodes in $R$ turns out more challenging, since any such path may still visit an unbounded number of edges in terms of $k$ and $r$.  We use two main ingredients to overcome these difficulties.  First, we apply our Theorem~\ref{thm:highprob} as a building block to obtain a solution with logarithmic congestion while losing only a constant factor in the approximation ratio.  Second, we introduce the concept of \emph{irreducible routings with low congestion} which allows us exploit the structural properties of the graph and the congestion property to identify a sufficiently large number of flow paths blocking only a small amount of flow.

Note that the natural greedy approach of always routing the shortest conflict-free path gives only $\mathcal{O}(\sqrt{m})$ for \textsc{MaxEDP}.  We believe that it is non-trivial to obtain our bounds via a more direct or purely combinatorial approach.

Our third result is a fixed-parameter algorithm for {\sc MaxNDP} in~$k + r$.
\begin{theorem}
\label{thm:maxndp-fpt-kr}
  {\sc MaxNDP} can be solved in time $(8k+8r)^{2r+2}\cdot \mathcal O(n)$ on graphs with feedback vertex set number~$r$ and $k$ terminal pairs.
\end{theorem}
This run time is polynomial for constant $r$.  We also note that for small~$r$, our algorithm is asymptotically significantly faster than the fastest known algorithm for NDP, by Kawarabayashi and Wollan~\cite{KawarabayashiWollan2010}, which requires time at least \emph{quadruple-exponential} in $k$ \cite{AdlerEtAl2011}.
Namely, if $r$ is at most triple-exponential in $k$, our algorithm is asymptotically faster than theirs. We achieve this result by the idea of so-called \emph{essential pairs} and \emph{realizations}, which characterizes the ``interaction'' between the feedback vertex set $R$ and the paths in an optimum solution.
Note that in our algorithm of Theorem~\ref{thm:maxndp-fpt-kr}, parameter $k$ does not appear in the exponent of the run time at all.
Hence, for small values of $r$ our algorithm is also faster than reducing {\sc MaxNDP} to {\sc NDP} by guessing the subset of pairs to be routed (at an expense of $2^k$ in the run time) and using Scheffler's~\cite{Scheffler1994} algorithm for {\sc NDP} with run time $2^{O(r\log r)}\cdot \mathcal O(n)$.

Once a fixed-parameter algorithm for a problem has been obtained, the existence of a polynomial-size kernel comes up.
Here we note that {\sc MaxNDP} does not admit a polynomial kernel for parameter $k+r$, unless~$\mathsf{NP}\subseteq \mathsf{coNP}/poly$~\cite{BodlaenderEtAl2011a}.

Another natural question is whether the run time $f(k,r)\cdot n$ in Theorem~\ref{thm:maxndp-fpt-kr} can be improved to~$f(r)\cdot n^{\mathcal O(1)}$.
We answer this question in the negative, ruling out the existence of a fixed-parameter algorithm for {\sc MaxNDP} parameterized by~$r$ (assuming $\mathsf{FPT}\not=\mathsf{W}[1]$):
\begin{theorem}
\label{thm:ndp-w1hard-r}
  {\sc MaxNDP} in unit-capacity graphs is $\mathsf{W}[1]$-hard parameterized by $r$.
\end{theorem}
This contrasts the known result that NDP is fixed-parameter tractable in~$r$~\cite{Scheffler1994}---which further stresses the relevance of understanding this parameter.

For {\sc MaxEDP}, we prove that the situation is, in a sense, even worse:
\begin{theorem}
\label{thm:edp-nphard-r2}
  {\sc MaxEDP} is $\mathsf{NP}$-hard for unit-capacity graphs with $r = 1$ and \textsc{EDP} is $\mathsf{NP}$-hard for unit-capacity graphs with $r=2$.
\end{theorem}
This theorem also shows that our algorithms are relevant for small values of $r$, and they nicely complement the $\mathsf{NP}$-hardness for {\sc MaxEDP} in capacitated trees~\cite{GargEtAl1997}.

Our results are summarized in Table~\ref{tab:results}.
\begin{table}[htpb]
  \centering
  \begin{tabular}{clcccc}
    \toprule
      const.     & param.    & EDP                      & {\sc MaxEDP}            & NDP & {\sc MaxNDP}\\
    \midrule
      $r = 0$    &           & poly~\cite{GargEtAl1997} & poly~\cite{GargEtAl1997} & poly~\cite{Scheffler1994} & poly~(Thm.~\ref{thm:maxndp-fpt-kr})\\      
      $r = 1$    &    & \emph{open}                              & {$\mathsf{NP}$-hard~(Thm.~\ref{thm:edp-nphard-r2})}          & poly~\cite{Scheffler1994} & poly~(Thm.~\ref{thm:maxndp-fpt-kr})\\
      $r \geq 2$ &           & $\mathsf{NP}$-hard~(Thm.~\ref{thm:edp-nphard-r2}) & $\mathsf{NP}$-hard~(Thm.~\ref{thm:edp-nphard-r2}) & poly~\cite{Scheffler1994} & poly~(Thm.~\ref{thm:maxndp-fpt-kr}) \\
      \midrule
                 & $r$       & \multicolumn{2}{c}{para-$\mathsf{NP}$-hard~(Thm.~\ref{thm:edp-nphard-r2})} & $\mathsf{FPT}$~\cite{Scheffler1994} & $\mathsf{W}[1]$-hard~(Thm.~\ref{thm:ndp-w1hard-r})\\
                 &           & \multicolumn{2}{c}{$\mathcal{O}(\sqrt{r}\cdot \log^{1.5}{kr})$-approx~(Thm.~\ref{thm:fvs})} & \multicolumn{2}{c}{exact $(k+r)^{\mathcal O(r)}$~(Thm.~\ref{thm:maxndp-fpt-kr})} \\
                 &           &\multicolumn{2}{c}{$\mathcal O(1)$-approx. w.cg. $\mathcal O\left(\frac{\log{kr}}{\log\log{kr}}\right)$} (Thm.~\ref{thm:highprob})\\
    \bottomrule
  \end{tabular}
  \caption{Summary of results obtained in this paper.}
  \label{tab:results}
\end{table}

\medskip
\noindent
\textbf{Related Work.}
Our study of the feedback vertex set number is in line with the general attempt to obtain bounds for MaxEDP (or related problems) that are independent of the input size.  Besides the above-mentioned works that provide bounds in terms of the \emph{tree-width} of the input graph, G{\"u}nl{\"u}k \cite{Gunluk2007} and Chekuri et al.~\cite{flow-cut-gaps-chekuri13} give bounds on the \emph{flow-cut gap} for the closely related integer multicommodity flow problem that are logarithmic with respect to the \emph{vertex cover number} of a graph. This improved upon earlier bounds of~$\mathcal O(\log n)$~\cite{LeightonRao1999} and $\mathcal O(\log k)$ \cite{flow-cut-theorem-aumann98,geometry-graphs-linial95}.
As every feedback vertex set is in particular a vertex cover of a graph, our results generalize earlier work for disjoint path problems on graphs with bounded vertex cover number.
Bodlaender et al.~\cite{BodlaenderEtAl2011a} showed that {\sc NDP} does not admit a polynomial kernel parameterized by vertex cover number \emph{and} the number~$k$ of terminal pairs, unless $\mathsf{NP}\subseteq \mathsf{coNP}/poly$ ; therefore, {\sc NDP} is unlikely to admit a polynomial kernel in $r + k$ either.
Ene et al.~\cite{EneEtAl2016} showed that {\sc MaxNDP} is $\mathsf{W}[1]$-hard parameterized by treedepth, which is another restriction of treewidth that is incomparable to the feedback vertex set number.

The basic gap in understanding the approximability of MaxEDP has led to several improved results for special graph classes, and also our results can be seen in this light. For example, polylogarithmic approximation algorithms are known for graphs whose global minimum cut value is $\Omega(\log^5 n)$ \cite{RaoZhou2010}, for bounded-degree
expanders~\cite{BroderEtAl1994,BroderEtAl1999,KleinbergRubinfeld1996,LeightonRao1999,Frieze2000}, and for Eulerian planar or 4-connected planar graphs~\cite{KawarabayashiKobayashi2011}.
Constant factor approximation algorithms are known for capacitated trees~\cite{GargEtAl1997,ChekuriEtAl2007}, grids and grid-like graphs~\cite{AumannRabani1995,AwerbuchEtAl1994,KleinbergTardos1995,KleinbergTardos1998}.
For planar graphs, there is a constant-factor approximation algorithm with congestion 2~\cite{SeguinCharbonneauShepherd2011}.
Very recently, Chuzhoy et al.~\cite{Chuzhoy-etal:NDP-planar} gave a $\tilde{\mathcal{O}}(n^{9/19})$-approximation algorithm for MaxNDP on \emph{planar} graphs. 
However, improving the $\mathcal O(\sqrt{n})$-approximation algorithm for \textsc{MaxEDP} remains elusive even for \emph{planar} graphs.

\section{Preliminaries}
\label{sec:preliminaries}
We use standard graph theoretic notation.
For a graph $G$, let $V(G)$ denote its vertex set and~$E(G)$ its edge set.
Let $G$ be a graph.
A \emph{feedback vertex set} of $G$ is a set $R\subseteq V(G)$ such that $G - R$ is a forest.
A \emph{minor} of~$G$ is a graph $H$ that is obtained by successively contracting edges from a subgraph of $G$ (and deleting any occurring loops).
A class $\mathcal G$ of graphs is \emph{minor-closed} if for any graph in $\mathcal G$ also all its minors belong to $\mathcal G$.

For an instance $(G,\mathcal M)$ of {\sc MaxEDP/MaxNPD}, we refer to the vertices participating in the pairs $\mathcal M$ as \emph{terminals}.
It is convenient to assume that~$\mathcal M$ forms a matching on the terminals; this can be ensured by making several copies of a terminal and attaching them as leaves.

\medskip
\noindent
\textbf{Multi-commodity flow relaxation.}
We use the following standard multi-commodity flow relaxation for {\sc MaxEDP} (there is an analogous relaxation for {\sc MaxNDP}).
We use $\mathcal P(u, v)$ to denote the set of all paths in $G$ from $u$ to $v$, for each pair $(u, v)$ of nodes. Since the pairs~$\mathcal M$ form a matching, the sets~$\mathcal P(s_i, t_i)$ are pairwise disjoint.
Let $\mathcal P = \bigcup_{i = 1}^k \mathcal P(s_i, t_i)$.
The LP has a variable $f(P)$ for each path $P \in \mathcal P$ representing the amount of flow on $P$.
For each pair $(s_i, t_i) \in \mathcal M$, the LP has a variable $x_i$ denoting the total amount of flow routed for the pair (in the corresponding IP, $x_i$ denotes whether the pair is routed or not).
The LP imposes the constraint that there is a flow from $s_i$ to $t_i$ of value $x_i$.
Additionally, the LP has constraints that ensure that the total amount of flow on paths using a given edge (resp. node for {\sc MaxNDP}) is at most 1. 

\begin{figure}[htb]
  \begin{center}
    \begin{boxedminipage}{0.38\textwidth}
      \vspace{-0.1in}
      \begin{align*}
	                     & (\textnormal{{\sc MaxEDP} LP})                           &\\
	   ~ \max \quad & \sum_{i = 1}^k x_i                                  &\\
	   \text{s.t.} \quad & \sum_{P \in \mathcal P(s_i, t_i)} f(P) = x_i \leq 1 & i=1,\hdots,k,~\\
                         & \sum_{P:\; e \in P} f(P) \leq 1       & e \in E(G)~\\\
                         & f(P) \geq 0                                         & P \in \mathcal P~
     \end{align*}
    \end{boxedminipage}
    \hspace{0.1in}
    \begin{boxedminipage}{0.4\textwidth}
      \centering
      \includegraphics[scale=0.63]{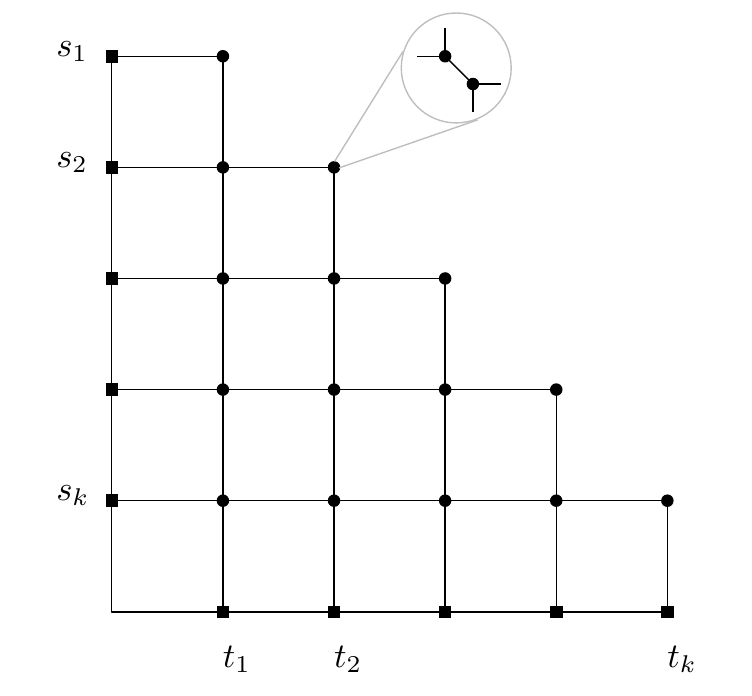}
    \end{boxedminipage}
  \end{center}
  \caption{Multi-commodity flow relaxation for {\sc MaxEDP}.
  Right: $\Omega(\sqrt{n})$ integrality gap for {\sc MaxEDP}~\cite{GargEtAl1997}: any integral routing routes at most one pair, whereas a multi-commodity flow can send $1/2$ unit of flow for each pair $(s_i, t_i)$ along the canonical path from $s_i$ to~$t_i$ in the grid.}
\label{fig:lp}
\end{figure}

It is well-known that the relaxation {\sc MaxEDP LP} can be solved in polynomial time, since there is an efficient separation oracle for the dual LP (alternatively, one can write a compact relaxation).
We use $(f, \mathbf{x})$ to denote a feasible solution to {\sc MaxEDP LP} for an instance $(G, \mathcal M)$ of {\sc MaxEDP}.
For each terminal $v$, let $x(v)$ denote the total amount of flow routed for $v$ and we refer to~$x(v)$ as the \emph{marginal value} of $v$ in the multi-commodity flow $f$. 

We will use the following result by Chekuri et al.~\cite[Sect. 3.1]{ChekuriEtAl2006}; see also Proposition 3.3 of Chekuri et al.~\cite{cns-tw-corr}. 
\begin{proposition}
\label{lem:single-node-routing}
  Let $(f,\bf{x})$ be a fractional solution to the LP relaxation of a \textsc{MaxEDP} instance $(G,\mathcal M)$.
  If some node $v$ is contained in all flow paths of~$f$, then we can find an integral routing of size at least $\frac{1}{12}\sum_{i}x_i$ in polynomial time.
\end{proposition}

\section{Bi-Criteria Approximation for MaxEDP with Low Congestion}
\label{sec:bi-crit-appr}
We present a randomized rounding algorithm that will lead to the proof of Theorem~\ref{thm:highprob}.

\subsection{Algorithm}
\label{sec:rand-rounding}
Consider an instance $(G,\mathcal M)$ of {\sc MaxEDP}.
Let $R$ be a 2-approximate minimum feedback vertex set of $G$ and let $r = |R|$; note that such a set $R$ can be obtained in polynomial time~\cite{BafnaEtAl1999}. 

For the sake of easier presentation, we will assume in this section that the feedback vertex set $R$ contains all terminal nodes from $\mathcal M$.  This can be achieved by temporarily adding the set of terminals to the feedback vertex set $R$.  Also note that this assumption increases the bound of Theorem~\ref{thm:highprob} by at most a constant factor.

First, solve the corresponding {\sc MaxEDP LP}.
We obtain an optimal solution~$(f,\mathbf{x})$. 
For each~$(s_i, t_i)\in\mathcal M$ we further obtain a set $\mathcal{P}'(s_i, t_i) = \{P \in \mathcal{P}(s_i, t_i) \mid~f(P)>0\}$ of positive weighted paths that satisfy the LP constraints.
Note that the total set ${\mathcal P' = \bigcup_{i = 1}^k \mathcal P'(s_i, t_i)}$ is of size polynomially bounded in the input size.
In what follows, we will modify~$\mathcal P'$ and then select an (unweighted) subset~$\mathcal S$ of~$\mathcal P'$ that will form our integral solution.

Each $P\in\mathcal{P'}$ has the form $(r_1,\dots,r_2,\dots,r_\ell)$ where $r_1,\dots,r_\ell$ are the nodes in~$R$ that are traversed by $P$ in this order.  The paths~$(r_j,\dots,r_{j+1})$ with $j=1,\dots,\ell-1$ are called \emph{subpaths of}~$P$.
For every subpath $P'$ of~$P$, we set~$f(P')=f(P)$.
Let~$\mathcal J$ be the multi-set of all subpaths of all paths in $\mathcal P'$.
Let $F = G-R$ be the forest obtained by removing~$R$.

We now modify some paths in~$\mathcal P'$, one by one, and at the same time construct a subset~$H$ of nodes that we will call ``hot spots''. 
At the end, every subpath in $\mathcal J$ will contain at least one hot spot.

Initially, let $H=\emptyset$.
Consider any tree~$T$ in $F$ and fix any of its nodes as a root. Then let~$\mathcal{J}_T$ be the multi-set of all subpaths in $\mathcal J$ that, excluding the endpoints, are contained in~$T$. For each subpath $P\in \mathcal{J}_T$, define its \emph{highest node}~$h(P)$ as the node on~$P$ closest to the root. Note that $P\cap T=P\cap F$ is a path.
Now, pick a subpath~$P\in \mathcal{J}_T$ that does not contain any node in $H$ and whose highest node $h(P)$ is \emph{farthest away} from the root. Consider the multi-set~$\mathcal{J}[P]$ of all subpaths in $\mathcal{J}_T$ that are identical to $P$ (but may be subpaths of different flow paths in $\mathcal{P}'$). Note that the weight~$f(\mathcal{J}[P]) := \sum_{P\in\mathcal{J}[P]} f(P)$ of $\mathcal{J}[P]$ is at most $1$ by the constraints of the LP. 
Let $u,v\in R$ be the endpoints of $P$. We define $\mathcal{J}_{uv}$ as the set of all subpaths in $\mathcal{J}\setminus \mathcal{J}[P]$ that have $u$ and $v$ as their endpoints and that do not contain any node in~$H$.

Intuitively speaking, we now aggregate flow on $P$ by rerouting as much flow as possible from~$\mathcal{J}_{uv}$ to~$P$.  To this end, we repeatedly perform the following operation as long as $f(\mathcal{J}[P])<1$ and $\mathcal{J}_{uv}\not=\emptyset$.  We pick a path~$P'$ in $\mathcal J$ that contains a subpath in $\mathcal{J}_{uv}$.  We reroute flow from $P'$ by creating a new path~$P''$ that arises from $P'$ by replacing its subpath between $u$ and $v$ with $P$, and assign it the weight $f(P'')=\min\{f(P'),1-f(\mathcal{J}[P])\}$.  Then we set the weight of (the original path) $P'$ to $\max\{0, f(P')+f(\mathcal{J}[P])-1\}$.  We update the sets~$\mathcal{P'}$,~$\mathcal{P}'(s_i,t_i)$,~$\mathcal{J}$,~$\mathcal{J}_T$, $\mathcal{J}[P]$ and $\mathcal{J}_{uv}$ accordingly.

As soon as $f(\mathcal{J}[P])=1$ or $\mathcal{J}_{uv}=\emptyset$, we add $h(P)$ to $H$. Then, we proceed with the next $P\in\mathcal{J}_T$ not containing a hot spot and whose highest node $h(P)$ is farthest away from the root. If no such $P$ is left we consider the next tree~$T$ in $F$. 

At the end, we create our solution~$\mathcal S$ by randomized rounding: 
We route every terminal pair~$(s_i, t_i)$ with probability $x_i$.  In case $(s_i,t_i)$ is routed, we randomly select a path from~$\mathcal P'(s_i, t_i)$ and add it to~$\mathcal S$ where the probability that path $P$ is taken is $f(P)/x_i$.

\subsection{Analysis}
First, observe that~$\mathbf{x}$ did not change during our modifications of the paths, as the total flow between any terminal pair did not change.
Thus, the expected number of pairs routed in our solution is~$\sum_{i=1}^k x_i\ge \opt$. Using the Chernoff bound, the probability that we route less than $\opt/2$ pairs is at most $e^{-1/8\opt}<1/2$, assuming that $\opt>8$.
Secondly, we bound the congestion of our solution---our second criterion.
\begin{lemma}
\label{thm:congestiononfisatmost2}
  The congestion of flow $f$ is at most~2.
\end{lemma}
\begin{proof}
In our algorithm, 
we increase the flow only along flow subpaths that are pairwise edge-disjoint.
To see this, consider two distinct flow subpaths $P$ and $P'$ on which we increase the flow.
Assume, without loss of generality, that $P$ was considered before $P'$ by the algorithm.
If there was an edge $e$ lying on~$P$ and $P'$, then both subpaths traverse the same tree in forest~$F$. 
Hence, the path from $e$ to $h(P')$ would visit $h(P)$, and~$h(P)$ would be an internal node of $P'$.  
This yields a contradiction, as~$h(P)$ was already marked as a hot spot when $P'$ was considered.  
This shows that we increased the flow along any edge by at most one unit, and, hence, $f$ has congestion at most~2.
\end{proof}

We now bound the congestion of the integral solution obtained by randomized rounding.  In the algorithm, we constructed a set $H$ of hot spots.  As a part of the analysis, we will now extend this set as follows.  We build a sub-forest $F'$ of $F$ consisting of all edges of $F$ that lie on a path connecting two hot spots.  Then we add to $H$ all nodes that have degree at least~3 in $F'$.  Since the number of nodes of degree~3 in any forest is at most its number of leaves and since every leaf of $F'$ is a hot spot, it follows that this can at most double the size of $H$.  Finally, we add the set $R$ of all feedback vertex nodes to $H$.
\begin{lemma}
\label{thm:onlyfewhotspots}
  The number $|H|$ of hot spots is $\mathcal O(kr^2)$.
\end{lemma}
\begin{proof}
  It suffices to show that the number of hot spots added to $H$ by the algorithm is $\mathcal O(kr^2)$.  To this end, fix two nodes $u,v\in R$ and consider the set of flow subpaths $P$ with end nodes $u$ and $v$ for which we added~$h(P)$ to $H$.  Due to the aggregation of flows in our algorithm, all except possibly one of the subpaths are saturated, that is, they carry precisely one unit of flow. Since no two of these subpaths are contained in a same flow path of $f$ and since the flow value of $f$ is bounded from above by $k$, we added only~$\mathcal O(k)$ hot spots for the pair $u,v$.  Since there are at most $r^2$ pairs in $R$, the claim follows.
\end{proof}
  
\begin{definition}
  A hot spot $u\in H$ is \emph{good} if the congestion on any edge incident on $u$ is bounded by $c\cdot \frac{\log{kr}}{\log\log{kr}}$, where $c$ is a sufficiently large constant; otherwise, $u$ is \emph{bad}.
\end{definition}

\begin{lemma}\label{lem:hotspot-prob}
  Let $u\in H$ be a hot spot. Then the probability that $u$ is bad is at most $1/(k^2r^3)$.
\end{lemma}
\begin{proof}
  Let $e_1=uv_1,\dots,e_{\ell}=uv_{\ell}$ be the edges incident on $u$ and let
  $f_i$ be the total flow on edge $uv_i$ for~$i=1,\dots,\ell$.  By Lemma~\ref{thm:congestiononfisatmost2}, we have that $f_i\leq 2$.
  Since any flow path visits at most two of the edges incident on~$u$, the total flow $\sum_{i=1}^{\ell}f_i$ on the edges incident on
  $u$ is at most $2k$.

  For any $i=1,\dots,\ell$, we have that
  $f_i=\sum_{P\colon P\ni{e_i}}f(P)$, where $P$ runs over the set of
  all paths connecting some terminal pair and containing $e_i$.  Let
  $f_{ij}=\sum_{P\in\mathcal{P}(s_j,t_j)\colon P\ni e_i}f(P)$ be the
  total amount of flow sent across~$e_i$ by terminal pair $(s_j,t_j)$.
  Recall that $x_j$ is the total flow sent for terminal pair $(s_j,t_j)$.  The
  probability that the randomized rounding procedure picks path $P$
  with $P\in\mathcal{P}(s_j,t_j)$ is precisely
  $x_j\cdot\frac{f(p)}{x_j}=f(p)$.  Given the disjointness of the
  respective events, the probability that pair $(s_j,t_j)$ routes a
  path across~$e_i$ is precisely~$f_{ij}$.  Let~$X_{ij}$ be the binary
  random variable indicating whether pair $(s_j,t_j)$ routes a path
  across~$e_i$. Then $\prob{X_{ij}=1}=f_{ij}$.  Let
  $X_i=\sum_{j}X_{ij}$ be the number of paths routed across $e_i$ by
  the algorithm.  By linearity of expectation, we have that
  $\expct{X_i}=\sum_{j}\expct{X_{ij}}=\sum_{j}f_{ij}=f_i$.

  Fix any edge~$e_i$. Set $\delta=c\cdot\frac{\log{kr}}{\log\log{kr}}$ and $\delta'=2\frac{\delta}{f_i}-1$.  Note that for fixed~$i$, the
  variables~$X_{ij}$ are independent.  Hence, by the Chernoff bound,
  we have that
    \begin{align*}
      \prob{X_i\geq c\cdot\frac{\log{kr}}{\log\log{kr}}} & \leq\prob{X_i\geq(1+\delta')f_i} < \left(\frac{e^{\delta'}}{(1+\delta')^{1+\delta'}}\right)^{f_i} \\
                                                         & \leq \left(\frac{f_i}{2}\right)^{2\delta}\cdot \left(\frac{\delta}{e}\right)^{-2\delta}
                                                         \leq f_i e^{-c'\log\log{kr}\cdot\frac{\log{kr}}{\log\log{kr}}}
                                                           \leq \frac{f_i}{2k^3r^3}\, .
    \end{align*}
Here, we use that $f_i\leq 2$ for the second last inequality and for the last inequality we pick $c'$ sufficiently large by making $c$ and $k$ sufficiently large. (Note that {\sc MaxEDP} can be solved efficiently for constant $k$.)

Now, using the union bound, we can infer that the probability that any of the edges incident on $u$ carries more than $\delta$ paths is at most $\sum_{i}f_i/(2k^3r^3)\leq (2k)/(2k^3r^3)=1/(k^2r^3)$.
\end{proof}

\begin{lemma}\label{lem:hot-spot-congestion}
  Assume that every hot spot is good. Then the congestion on any edge is bounded by $2c\frac{\log{kr}}{\log\log{kr}}$.
\end{lemma}
\begin{proof}
  Consider an arbitrary edge $e=uv$ that is not incident on any hot
  spot.  In particular, this means that~$e$ lies in the forest
  $F=G-R$.  A hot spot $z$ in $F$ is called \emph{direct} to $u$ (or
  $v$) if the path in $F$ from $z$ to $u$ (or $v$) neither contains
  $e$ nor any hot spot other than $z$.

  Now observe that there can be only one hot spot $z$ direct to $u$
  and only one hot spot $z'$ direct to $v$.  If there was a second hot
  spot $z''\neq z$ direct to $u$ then there would have to be yet
  another hot spot at the node where the path $P_z$ from $z$ to $u$
  joins the path from $z''$ to $u$ contradicting the choice of $z$.
  Let $P_{z'}$ be the path from $z'$ to $v$ in $F$.  Moreover, let
  $e_z$ be the edge incident on $z$ on path $P_z$ and let $e_{z'}$ be
  the edge incident on~$z'$ on path $P_{z'}$.

  Now let $P$ be an arbitrary path that is routed by our algorithm and
  that traverses $e$.  It must visit a hot spot.  If $P$ visited neither
  $z$ nor $z'$, then $P$ would contain a hot spot direct to $u$
  or to $v$ that is distinct from~$z$ and $z'$---a contradiction.  Therefore, $P$
  contains $e_z$ or $e_z'$.  The claim now follows from the
  fact that this holds for any path traversing $e$, that $z$ and $z'$
  are good, and that therefore at most
  $2c\frac{\log{kr}}{\log\log{kr}}$ paths visit $e_z$ or~$e_z'$.
\end{proof}

\begin{theorem}
  The algorithm from Sect.~\ref{sec:rand-rounding} produces---with constant probability---a routing with $\Omega(\opt)$ paths, such that the congestion is $\mathcal O\left(\frac{\log{kr}}{\log\log{kr}}\right)$.
\end{theorem}
\begin{proof}
  As argued above, we route less than $\opt/2$ paths with probability at most $1/2$.  By Lemma~\ref{thm:onlyfewhotspots}, there are~$\mathcal O(kr^2)$ hotspots.  The probability that at least one of these hot spots is bad is ${\mathcal O(kr^2/(k^2r^3))=\mathcal O(1/(kr))}$, by Lemma~\ref{lem:hotspot-prob}.  Hence, with constant probability, we route at least $\opt/2$ pairs with congestion at most~$2c\frac{\log{kr}}{\log\log{kr}}$, by Lemma~\ref{lem:hot-spot-congestion}.
\end{proof}

\section{Refined Approximation Bound for MaxEDP}\label{sec:approx-algo-maxedp}
In this section, we provide an improved approximation guarantee for {\sc MaxEDP} \emph{without} congestion, thereby proving Theorem~\ref{thm:fvs}. (In contrast to the previous section, we do not assume here that all terminals are contained in the feedback vertex set.)

\subsection{Irreducible Routings with Low Congestion}\label{sec:double-path-cover}

We first develop the concept of \emph{irreducible routings with low congestion}, which is (besides Theorem~\ref{thm:highprob}) a key ingredient of our strengthened bound on the approximability of \textsc{MaxEDP} based on the feedback vertex number.

Consider any multigraph~$G$ and any set~$\mathcal P$ of (not necessarily simple) paths in~$G$ with congestion~$c$. 
We say that an edge~$e$ is \emph{redundant in $\mathcal{P}$} if there is an edge~$e'\neq e$ such that the set of paths in $\mathcal P$ \emph{covering} (containing) $e$ is a subset of the set of paths in~$\mathcal P$ covering~$e'$.

\begin{definition}
  Set $\mathcal P$ is called an \emph{irreducible routing
    with congestion $c$} if each edge belongs to at most $c$
  paths of $\mathcal P$ and there is no edge redundant in~$\mathcal P$.
\end{definition}
In contrast to a feasible routing of an {\sc MaxEDP} instance, we do not require an irreducible routing to connect a set of terminal pairs.
If there is an edge $e$ redundant in $\mathcal{P}$, we can apply the following \emph{reduction rule}:
We contract~$e$ in $G$ and we contract~$e$ in every path of $\mathcal P$ that covers~$e$.
By this, we obtain a minor~$G'$ of~$G$ and a set~$\mathcal P'$ of paths that consists of all the contracted paths and of all paths in $\mathcal P$ that were not contracted.
Thus, there is a one-to-one correspondence between the paths in~$\mathcal P$ and~$\mathcal P'$ .

We make the following observation about $\mathcal{P}$ and $\mathcal{P}'$.
\begin{observation}\label{obs:redToDPC}
Any subset of paths in $\mathcal P'$ is edge-disjoint in $G'$ if and only if the corresponding subset of paths in $\mathcal P$ is edge-disjoint in~$G$.
\end{observation}

Since the application of the reduction rule strictly decreases the number of redundant edges, an iterative application of this rule yields an irreducible routing on a minor of the original graph.
\begin{theorem}\label{thm:double-path-covers}
  Let $\mathcal G$ be a minor-closed class of multigraphs and let $p_{\mathcal G} > 0$.
  If for each graph $G\in\mathcal G$ and every non-empty irreducible routing $\mathcal S$ of $G$ with congestion~$c$ there exists a path in~$\mathcal S$ of length at most~$p_{\mathcal G}$, then the average length of the paths in $\mathcal S$ is at most $c\cdot p_{\mathcal G}$.
\end{theorem}
\begin{proof}
Take a path~$P_0$ of length at most $p_{\mathcal G}$. Contract all edges of $P_0$ in $G$ and obtain a minor $G'\in\mathcal G$ of $G$.  For each path in $\mathcal S$ contract all edges shared with $P_0$ to obtain a set~$\mathcal S'$ of paths. 
Remove $P_0$ along with all degenerated paths from~$\mathcal S'$, thus $|\mathcal S'| < |\mathcal S|$.   
Note that~$\mathcal S'$ is an irreducible routing of $G'$ with congestion~$c$.
We repeat this reduction procedure recursively on~$G'$ and $S'$ until $S'$ is empty which happens after at most $|\mathcal{S}|$ steps.
At each step we decrease the total path length by at most $c\cdot p_{\mathcal G}$. Hence, the total length of paths in $\mathcal{S}$ is at most~$|\mathcal{S}| \cdot c \cdot p_{\mathcal G}$.
\end{proof}

As a consequence of Theorem~\ref{thm:double-path-covers}, we get the following result for forests.
\begin{lemma}\label{lem:path-length-forest}
  Let $F$ be a forest and let $\mathcal{S}$ be a non-empty irreducible routing of $F$ with congestion~$c$.  Then the average path length in~$\mathcal{S}$ is at most $2c$. 
\end{lemma}
\begin{proof}
 We show that $\mathcal{S}$ contains a path of length as most~$2$. The lemma follows immediately by applying Theorem~\ref{thm:double-path-covers}.

  Take any tree in $F$, root it with any node and consider a leaf~$v$ of maximum depth. 
  Let~$e_1$ and $e_2$ be the first two edges on the path from $v$ to the root.
  By definition of irreducible routing, 
  the set of all paths covering $e_1$ is not a subset of the paths covering~$e_2$, hence, 
  $e_1$ is covered by a path which does not cover~$e_2$.
  Since all other edges incident to~$e_1$ end in a leaf, this path has length at most~$2$.
\end{proof}

Note that the bound provided in Lemma~\ref{lem:path-length-forest} is actually tight up to a constant. 
Let $c\geq 1$ be an arbitary integer. Consider a graph that is a path of length~$c-1$ with a star of $c-1$ leafs attached to one of its end points.  
The $c-1$ many paths of length~$c$ together with the $2c-2$ many paths of length~$1$ form an irreducible routing with congestion $c$. The average path length is $((c-1)c+(2c-2))/(3c-3)=(c+2)/3$.

\subsection{Approximation Algorithm}\label{sec:bound-size-feedb}
Consider an instance $(G,\mathcal M)$ of \textsc{MaxEDP}, and let $r$ be the size of a feedback vertex set $R$ in $G$.  Using our result of Sect.~\ref{sec:bi-crit-appr}, we can efficiently compute a routing $\mathcal{P}$ with congestion~$c\coloneqq \mathcal O\left(\frac{\log{kr}}{\log\log{kr}}\right)$ containing~$\Omega(\opt)$ paths.

Below we argue how to use the routing $\mathcal P$ to obtain a feasible routing of cardinality $\Omega\left(|\mathcal{P}|/(c^{1.5} \sqrt{r})\right)$, which yields an overall approximation ratio of~$\mathcal{O}\left(\sqrt{r}\cdot\log^{1.5} rk\right)$; that will prove Theorem~\ref{thm:fvs}.

Let $r'=\sqrt{r/c}$. We distinguish the following cases. 

\medskip
\noindent
\textbf{Case 1:} At least half of the paths in $\mathcal P$ visit at most $r'$ nodes of the feedback vertex set $R$. 
Let $\overline{\mathcal P}$ be the subset of these paths.
As long as there is an edge~$e$ not adjacent to~$R$ that is redundant in~$\mathcal P'$,
we iteratively apply the reduction rule from Sect.~\ref{sec:double-path-cover} on~$e$.
Let $G'$ be the obtained minor of $G$ with forest $F'=G'-R$, and let $\mathcal P'$ be the obtained set of (not necessarily simple) paths corresponding to~$\overline{ \mathcal P}$.
By Observation~\ref{obs:redToDPC}, it suffices to show that there is a subset ${\mathcal{P}}_0'\subseteq{\mathcal{P}'}$ of pairwise edge-disjoint paths of size~$|\mathcal P_0| = \Omega\left(|\mathcal{P}|/( c r') \right)$
in order to obtain a feasible routing for~$(G,\mathcal M)$ of size~$\Omega\left(|\mathcal{P}|/( c r') \right)$.

To obtain~$\mathcal{P}_0'$,  
we first bound the total path length in~$\mathcal P'$.
Removing~${R}$ from ${G'}$ ``decomposes'' the set ${P}'$ into a set~${\mathcal{S}:=\{S\textrm{ is a connected component of }P\cap F\mid P\in\mathcal{P}'\,\}}$ of 
subpaths lying in $F'$. 
Observe that $\mathcal{S}$ is an irreducible set of~$F'$ with congestion~$c$, as the reduction rule is not applicable anymore. 
(Note that a single path in~$\mathcal{P'}$ may lead to many paths in the cover~$\mathcal{S}$ which are considered distinct.) 
Thus, by Lemma~\ref{lem:path-length-forest}, the average path length in $\mathcal{S}$ is at most~$2c$.

Let $P$ be an arbitrary path in $\mathcal{P}'$.  Each edge on $P$ that is \emph{not} in a subpath in~$\mathcal{S}$ is incident on a node in~${R}$, and each node in ${R}$ is incident on at most two edges in~$P$.  Together with the fact that $P$ visits at most~$r'$ nodes in $R$ and 
that the average length of the subpaths in $\mathcal{S}$ is at most~$2c$,
we can upper bound the total path length $\sum_{P\in\mathcal{P'}}|P|$ by~$|\mathcal{P'}| r' (2c+2)$.
Let $\mathcal P''$ be the set of the $|\mathcal{P}'|/2$ shortest paths in~$\mathcal{P}'$.
Hence, each path in~$\mathcal P''$ has length at most~$4 r' (c+1)$.   

We greedily construct a feasible solution $\mathcal{P}_0'$ by iteratively picking an arbitrary path~$P$ from~$\mathcal{P}''$ adding it to $\mathcal{P}_0'$ and removing all paths from $\mathcal{P}''$ that share some edge with~$P$ (including $P$ itself).  We stop when~$\mathcal{P}''$ is empty.  As $\mathcal{P}''$ has congestion~$c$, 
we remove at most~$4r'c(c+1)$ paths from~$\mathcal{P}''$ per iteration. Thus,~$|\mathcal{P}_0'|\geq|\mathcal{P''}|/(4r'c(c+1))=\Omega\left(|\mathcal{P}|/( c^{1.5} \sqrt{r}\right)$.

\medskip
\noindent
\textbf{Case 2:} At least half of the paths in $\mathcal P$ visit at least $r'$ nodes of the feedback vertex set $R$.
Let $\mathcal P'$ be the subset of these paths.
Consider each path in~$\mathcal{P}'$ as a flow of value $1/c$ and let~$f$ be the sum of all these flows. 
Note that $f$ provides a feasible solution to the \textsc{MaxEDP} LP relaxation for $({G}, M)$ of value at least~$|{\mathcal{P}}|/(2c)$. 
Note that each such flow path contributes~$1/c$ unit of flow to each of the $r'$ nodes in ${R}$ it visits.  
Since every flow path in~$f$ has length at least $r'$, the total inflow of the nodes in ${R}$ is at least $|f|r'$.  
By averaging, there must be a node $v\in {R}$ of inflow at least $r'|f|/r = |f|/r'$.
Let $f'$ be the subflow of $f$ consisting of all flow paths visiting $v$.
This subflow corresponds to a feasible solution $(f',\bf{x'})$ of the LP relaxation of value at least 
$|f|/r'\geq |\mathcal{P}|/(2cr')$. 
Using Proposition~\ref{lem:single-node-routing}, we can recover an integral feasible routing of size at least~$\frac{1}{12}\sum_{i}x_i'\geq |{\mathcal{P}}|/(24cr')
=\Omega\left(|\mathcal{P}|/( c^{1.5} \sqrt{r}\right)$. 

This completes the proof of~Theorem~\ref{thm:fvs}.
 \qed

\section{Fixed-Parameter Algorithm for MaxNDP}
\label{sec:exact-ndp}

We give a fixed-parameter algorithm for {\sc MaxNDP} with run time $(k+r)^{\mathcal O(r)}\cdot n$, where $r$ is the size of a minimum feedback vertex set in the given instance~$(G,\mathcal M)$.
A feedback vertex set $R$ of size $r$ can be computed in time $2^{O(r)}\cdot \mathcal O(n)$~\cite{LokshtanovEtAl2015}.
By the matching assumption, each terminal in~$\mathcal M$ is a leaf.  We can thus assume that none of the terminals is contained in~$R$.

Consider an optimal routing $\mathcal{P}$ of the given {\sc MaxNDP} instance. Let $\mathcal{M}_R\subseteq\mathcal{M}$ be the set of terminal pairs that are connected via $\mathcal{P}$ by a path that visits at least one node in~$R$.  Let $P\in\mathcal{P}$ be a path connecting a terminal pair $(s_i,t_i)\in\mathcal{M}_R$.  This path has the form $(s_i,\dots,r_1,\dots,r_2,\dots,r_\ell,\dots,t_i)$, where $r_1,\dots,r_\ell$ are the nodes in $R$ that are traversed by $P$ in this order.  The pairs $(s_i,r_1),(r_\ell,t_i)$ and~$(r_j,r_{j+1})$ with $j=1,\dots,\ell-1$ are called \emph{essential} pairs for $P$.  A node pair is called \emph{essential} if it is essential for some path in $\mathcal{P}$.  Let $\mathcal{M}_e$ be the set of essential pairs.

Let $F$ be the forest that arises when deleting $R$ from the input graph $G$.  Let~$(u,v)$ be an essential pair.  A $u$-$v$ path $P$ in $G$ is said to \emph{realize} $(u,v)$ if all internal nodes of $P$ lie in $F$.  A set $\mathcal{P}'$ of paths is said to \emph{realize}~$\mathcal{M}_e$ if every pair in $\mathcal{M}_e$ is realized by some path in $\mathcal{P}'$ and if two paths in $\mathcal{P}'$ can only intersect at their end nodes. Note that the optimal routing $\mathcal{P}$ induces a natural realization of~$\mathcal{M}_e$, by considering all maximal subpaths of paths in $\mathcal{P}$ whose internal nodes all lie in~$F$.  
Conversely, for any realization~$\mathcal{P}'$ of~$\mathcal{M}_e$, we can concatenate paths in~$\mathcal{P}'$ to obtain a feasible routing that connects all terminal pairs in~$\mathcal{M}_R$. Therefore, we consider $\mathcal{P}'$ (slightly abusing notation) also as a feasible routing for $\mathcal{M}_R$.

In our algorithm, we first guess the set~$\mathcal{M}_e$ (and thus $\mathcal{M}_R$). 
Then, by a dynamic program, we construct two sets of paths, $\mathcal{P}_e$ and $\mathcal{P}_F$ 
where $\mathcal{P}_e$ realizes~$\mathcal{M}_e$ and~$\mathcal{P}_F$ connects in~$F$ a subset of $\overline{\mathcal{M}}_R:=\mathcal{M}\setminus\mathcal{M}_R$. In our algorithm, the set $\mathcal{P}_e\cup\mathcal{P}_F$ forms a feasible routing that maximizes $|\mathcal{P}_F|$ and routes all pairs in~$\mathcal{M}_R$.  (Recall that we consider the realization $\mathcal{P}_e$ of~$\mathcal{M}_e$ as a feasible routing for~$\mathcal{M}_R$.)

Now assume that we know set $\mathcal{M}_e$. We will describe below a dynamic program that computes an optimum routing in time $2^{\mathcal O(r)}(k+r)^{\mathcal O(1)}n$.  For the sake of easier presentation, we only describe how to compute the cardinality of such a routing.

We make several technical assumptions that help to simplify the presentation.  First, we modify the input instance as follows. We subdivide every edge incident on a node in $R$ by introducing a single new node on this edge.  Note that this yields an instance equivalent to the input instance.  As a result, every neighbor of a node in $R$ that lies in $F$, that is, every node in $N_G(R)$, is a leaf in~$F$.  
Moreover, the set~$R$ is an independent set in $G$.  Also recall that we assumed that every terminal is a leaf. Therefore, we may assume that $R$ does not contain any terminal.  We also assume that forest $F$ is a rooted tree, by introducing a dummy node (which plays the role of the root) and arbitrarily connecting this node to every connected component of $F$ by an edge.  In our dynamic program, we will take care that no path visits this root node. We also assume that $F$ is an ordered tree by introducing an arbitrary order among the children of each node.

For any node $v$, let $F_v$ be the subtree of $F$ rooted at~$v$. 
Let $\ch{v} := \deg_F(v)-1$ be the number of children of~$v$ 
and let $v_1,\dots v_{\ch{v}}$ be the (ordered) children of~$v$.  
Then, for $i=1,\dots,\ch{v}$, let~$F_v^i$ denote the subtree of~$F_v$ induced by the union of~$v$ with the subtrees $F_{v_1},\dots,F_{v_i}$.
For leaves~$v$, we define $F_v^0$ as $F_v=v$. 

We introduce a dynamic programming table $T$.  
It contains an entry for every~$F_v^i$ and every subset~$\mathcal{M}_e'$ of~$\mathcal{M}_e$. 
Roughly speaking, the value of such an entry is the solution to the subproblem, where we restrict the forest to~$F_v^i$, and the set of essential pairs to~$\mathcal{M}_e'$.  
More precisely, table~$T$ contains five parameters. Parameters~$v$ and $i$ describing $F_v^i$, parameter~$\mathcal{M}_e'$, and two more parameters~$u$ and~$b$. 
Parameter 
$u$ is either a terminal,
or a node in $R$, 
and~$b$ is in one of the three states: $\free$, $\toBeUsed$, or $\blocked$. 
The value~$T[v,i,\mathcal{M}_e',u,b]$ is the maximum cardinality of a set $\mathcal{P}_F$ of paths with the following properties:
\begin{enumerate}
\item\vspace{-0.5em} $\mathcal{P}_F$ is a feasible routing of some subset of $\overline{\mathcal{M}}_R$.
\item $\mathcal{P}_F$ is completely contained in $F_v^i$. 
\item There is an additional set $\mathcal{P}_e$ of paths with the following properties:
\begin{enumerate}
\item $\mathcal{P}_e$ is completely contained in $F_v^i\cup R$ and node-disjoint from the paths in $\mathcal{P}_F$.
\item $\mathcal{P}_e$ is a realization of $\mathcal{M}_e'\cup\{(u,v)\}$ if $b=\toBeUsed$. 
Else, it is a realization of~$\mathcal{M}_e'$. 
\item There is no path in $\mathcal{P}_e\cup\mathcal{P}_F$ visiting $v$ if $b=\free$. 
\end{enumerate}
\end{enumerate}
If no such set $\mathcal{P}_F$ exists then $T[v,i,\mathcal{M}_e',u,b]$ is $-\infty$.  

Note that the parameter~$u$ is only relevant when $b=\toBeUsed$ (otherwise, it can just be ignored). 
Observe that~$T[v,i,\mathcal{M}_e',u,\blocked]\ge T[v,i,\mathcal{M}_e',u,\free]\ge T[v,i,\mathcal{M}_e',u,\toBeUsed]$.
Below, we describe how to compute the entries of $T$ in a bottom-up manner.

In the base case $v$ is a leaf. We set $T[v,0,\emptyset,u,\free]=0$. 
Then we set~$T[v,0,\mathcal{M}_e',u,\blocked]=0$ if $\mathcal{M}_e'$ is either empty, consists of a single pair of nodes in~$R\cap N_G(v)$, 
or consists of a single pair where one node is~$v$ and the other one is in~$R\cap N_G(v)$.
Finally, we set $T[v,0,\emptyset,u,\toBeUsed]=0$ if $u=v$ or $u$ is in~$R\cap N_G(v)$. 
For all other cases where $v$ is a leaf, we set $T[v,i,\mathcal{M}_e',u,b]=-\infty$.

For the inductive step, we consider the two cases $i=1$ and $i>1$. 
Let~$i=1$. It holds that~$T[v,1,\mathcal{M}_e',u,\toBeUsed]=T[v_1,\ch{v},\mathcal{M}_e',u,\toBeUsed]$ since the path in $\mathcal{P}_e$ realizing $(u,v)$ has to start at a leaf node of $F_{v_1}$.
It also holds that $T[v,1,\mathcal{M}_e',u,\blocked]$ and $T[v,1,\mathcal{M}_e',u,\free]$ are 
equal to~$T[v_1,\ch{v},\mathcal{M}_e',u,\blocked]$.

Now, let $i>1$. In a high level view, we guess which part of $\mathcal{M}_e'$ is realized in $F_{v}^{i-1}\cup R$ and which part is realized in $F_{v_i}\cup R$.
For this,
we consider every tuple~$(\mathcal{M}_{e1}',\mathcal{M}_{e2}')$ such that $\mathcal{M}_{e1}' \uplus \mathcal{M}_{e2}'$ is a partition of~$\mathcal{M}_e'$. By our dynamic programming table, we find a tuple that maximizes our objective.
In the following, we assume that we guessed~$(\mathcal{M}_{e1}',\mathcal{M}_{e2}')$ correctly. Let us consider the different cases of $b$ in more detail.

For $b=\free$, node $v$ is not allowed to be visited by any path, especially by any path in~$F_{v}^{i-1}\cup R$. Hence,~$T[v,i,\mathcal{M}_e',u,\free]$ is equal to \[T[v,i-1,\mathcal{M}_{e1}',u,\free] + T[v_i,\ch{v_i},\mathcal{M}_{e2}',u,\blocked]~.\] 

In the case of $b=\toBeUsed$, we have to realize $(u,v)$ in~$F_{v}^{i}\cup R$. For this, there are two possibilities: Either $(u,v)$ is realized by a path 
in~$F_{v}^{i-1}\cup R$, or there is a realizing path that first goes through~$F_{v_i}\cup R$ and then reaches $v$ via the edge $(v_i, v)$. 
Hence, for the first case, we consider \[T[v,i-1,\mathcal{M}_{e1}',u,\toBeUsed] + T[v_i,\ch{v_i},\mathcal{M}_{e2}',u,\blocked],\] 
for the second case, we consider \[T[v,i-1,\mathcal{M}_{e1}',u,\free] + T[v_i,\ch{v_i},\mathcal{M}_{e2}',u,\toBeUsed]~.\] 
Maximizing over both, 
we obtain~$T[v,i,\mathcal{M}_e',u,\toBeUsed]$. 

For the case of $b=\blocked$, we will consider two subcases.
In the first subcase, there is no path in~$\mathcal{P}_e\cup\mathcal{P}_F$ going through edge $(v_i,v)$, hence, we get
\[T[v,i-1,\mathcal{M}_{e1}',u,\blocked] + T[v_i,\ch{v_i},\mathcal{M}_{e2}',u,\blocked]~.\]
In the second subcase, there is a path~$P$ in $\mathcal{P}_e\cup\mathcal{P}_F$ going through edge $(v_i,v)$. 
Since $P$ is connecting two leafs in~$F_{v}^{i}$, a part of $P$ is in~$F_{v}^{i-1}\cup R$ and the other part is in~$F_{v_i}\cup R$.
If $P\in\mathcal{P}_e$, then it is realizing a pair of~$\mathcal{M}_{e}'$. Hence, for every pair~$(u_1,u_2)\in\mathcal{M}_{e}'$, we have to consider the term
\[T[v,i-1,\mathcal{M}_{e1}'-(u_1,u_2),u_1,\toBeUsed] + T[v_i,\ch{v_i},\mathcal{M}_{e2}'-(u_1,u_2),u_2,\toBeUsed]\]
and the symmetric term where we swap $u_1$ and $u_2$.
If $P\in\mathcal{P}_F$, then it is realizing a terminal pair of~$\overline{\mathcal{M}}_R$. Hence, for every pair~$(u_1,u_2)\in\overline{\mathcal{M}}_R$ we get the term 
\[1+ T[v,i-1,\mathcal{M}_{e1}',u_1,\toBeUsed] + T[v_i,\ch{v_i},\mathcal{M}_{e2}',u_2,\toBeUsed]\]
and the symmetric term where we swap $u_1$ and $u_2$. Note that we count the path realizing~$(u_1,u_2)$ in our objective. 
Maximizing over all the terms of the two subcases, we obtain~$T[v,i,\mathcal{M}_e',u,\toBeUsed]$. 

Let us analyze the run time of algorithm described in Sect.~\ref{sec:exact-ndp}.
In order to guess $\mathcal{M}_e$, we enumerate all potential sets of essential pairs. 
There are at most $(2k+r+1)^{2r}$ candidate sets to consider, since each pair contains a node in~$R$, and each node in $R$ is paired with at most two other nodes each of which is either a terminal or another node in~$R$.  
For each particular guess $\mathcal{M}_e$, we run the above dynamic program. 
The number of entries in $T$---as specified by the five parameters $v$, $i$, $\mathcal{M}_e', u$ and $b$---for each fixed $\mathcal{M}_e$ is at most~$(\sum_{v\in V(F)}\deg_F(v))\times 2^{2r}\times(2k+r)\times 3$.  
For the computation of each such entry, we consider all combinations of at most~$2^{2r}$ partitions of $\mathcal{M}_e'$ with either at most $r$ essential pairs in~$\mathcal{M}_e'$, or with at most~$k$ terminal pairs in~$\overline{\mathcal{M}}_R$.  Altogether, this gives a run time of~$(8k+8r)^{2r+2}\cdot \mathcal O(n)$.
This finishes the proof of Theorem~\ref{thm:maxndp-fpt-kr}.

\section{Parameterized Intractability of {\sc MaxNDP} for Parameter $r$}
\label{sec:ndp-w1hard-r}
In this section we show that {\sc MaxNDP} is $\mathsf{W}[1]$-hard parameterized by the size~$r$ of a feedback vertex set.
This reduction was originally devised for parameter treedepth, by Ene et al.~\cite{EneEtAl2016}; here we notice that the same reduction also works for parameter $r$.
(Both treedepth and feedback vertex set number are restrictions of treewidth, but they are incomparable to each other.)

For sake of completeness, we include the reduction here, and argue about the feedback vertex set number of the reduced graph.
The reduction is from the~$W[1]$-hard \textsc{Multicolored Clique} problem~\cite{FellowsEtAl2009}, where given a graph
$G$, an integer~$k$, and a partition $V = V^1 \uplus V^2 \uplus \ldots \uplus V^k$,
we are to check if there exists $k$-clique in $G$ with exactly one
vertex in every set~$V^i$. By adding dummy vertices, we can assume that $|V^i| = n$ for every $i=1,\hdots,k$, and that $n, k \geq 2$.

\medskip
\noindent
\textbf{Construction.}
Given an instance $(G,k,(V^i)_{i=1}^k)$ of \textsc{Multicolored Clique}, we aim at constructing an equivalent instance $(H,\mathcal M,\ell)$ of {\sc MaxNDP}.

We start with a construction, for every set $V^i$, a gadget $W^i$ as follows.
First, for every $v \in V^i$ we construct a $(k-1)$-vertex path $X_v^i$ on vertices~$x_{v,1}^i, x_{v,2}^i, \ldots,  x_{v,i-1}^i, x_{v,i+1}^i, \ldots, x_{v,k}^i$.
Second, we select an arbitrary vertex~$u^i \in V_i$.
Third, for every $v \in V^i \setminus \{u^i\}$, we add a vertex~$s^i_v$ adjacent to the first vertex of~$X_v^i$ (i.e., $x_{v,1}^i$ and $x_{u^i,1}^i$ if $i > 1$ or $x_{v,2}^i$ and $x_{u^1,2}^i$ if $i=1$),
  a vertex~$t^i_v$ adjacent to the last vertex of~$X_v^i$ (i.e., $x_{v,k}^i$ and $x_{u^i,k}^i$ if $i < k$ or~$x_{v,k-1}^i$ and $x_{u^i,k-1}^i$ if $i=k$), and make $(s^i_v,t^i_v)$ a terminal pair.
This concludes the description of the gadget $W^i$. By $\mathcal M_{st}$ we denote the set of terminal pairs constructed in this step.

To encode adjacencies in $G$, we proceed as follows. For every pair $1 \leq i < j \leq k$, we add a vertex $p_{i,j}$, adjacent
to all vertices $x_{v,j}^i$ for $v \in V_i$ and all vertices $x_{u,i}^j$ for $u \in V_j$. For every edge $vu \in E(G)$ with $v \in V_i$ and $u \in V_j$,
we add a terminal pair $(x_{v,j}^i, x_{u,i}^j)$. Let $\mathcal M_x$ be the set of terminal pairs constructed in this step; we have $\mathcal M = \mathcal M_{st} \cup \mathcal M_x$.

Finally, we set the required number of paths $\ell := k(n-1) + \binom{k}{2}$. This concludes the description of the instance $(H,\mathcal M,\ell)$.

\medskip
\noindent
\textbf{From a clique to disjoint paths.}
Assume that the input \textsc{Multicolored Clique} instance is a ``yes''-instance, and let $\{v^i~|~i = 1,\hdots,k\}$ be a clique in $G$ with $v^i \in V^i$ for $i = 1,\hdots,k$.
We construct a family of~$\ell$ vertex-disjoint paths as follows. First, for $i = 1,\hdots,k$ and every $v \in V^i \setminus \{u^i\}$, we route a path from~$s^i_v$ to $t^i_v$
through the path $X_v^i$ if $v \neq v^i$, and through the path $X_{u^i}^i$ if $v = v^i$. Note that in this step we have created~$k(n-1)$ vertex-disjoint paths
connecting terminal pairs, and in every gadget $W^i$ the only unused vertices are vertices on the path~$X_{v^i}^i$.
To construct the remaining $\binom{k}{2}$ paths, for every pair $1 \leq i < j \leq k$ we take the $3$-vertex path from $x_{v^i,j}^i$ to $x_{v^j,i}^j$ through $p_{i,j}$;
note that the assumption that $v^iv^j \in E(G)$ ensures that $(x_{v^i,j}^i, x_{v^j,i}^j)$ is indeed a terminal pair in $\mathcal M$.

\medskip
\noindent
\textbf{From disjoint paths to a clique.}
In the other direction, let $\mathcal P$ be a family of $\ell$ vertex-disjoint paths connecting terminal pairs in $H$. 
Let $\mathcal P_{st} \subseteq \mathcal P$ be the set of paths connecting terminal pairs from $\mathcal M_{st}$, and similarly define $\mathcal P_x$.
First, observe that the set $P = \{p_{i,j}~|~1 \leq i < j \leq k\}$ separates every terminal pair from~$\mathcal M_x$. Hence, every path from $\mathcal P_x$ contains at least one vertex from $P$.
Since $|P| = \binom{k}{2}$, we have~$|\mathcal M_x| \leq \binom{k}{2}$, and, consequently,~$|\mathcal P_{st}| \geq \ell - \binom{k}{2} = k(n-1) = |\mathcal M_{st}|$. 
We infer that $\mathcal P_{st}$ routes all terminal pairs in $\mathcal M_{st}$ without using any vertex of $P$, while $\mathcal P_x$ routes $\binom{k}{2}$ pairs from $\mathcal P_x$, and every path from $\mathcal P_x$ contains exactly one vertex from $P$.

Since the paths in $\mathcal P_{st}$ cannot use any vertex in $P$, every such path needs to be contained inside one gadget~$W^i$. Furthermore, observe that a shortest path between terminals $s_{v,a}^i$ and $t_{v,a}^i$ inside $W^i$ is either~$X_{u^i}^i$ or~$X_v^i$,
      prolonged with the terminals at endpoints, and thus contains $k+1$ vertices.
Furthermore, a shortest path between two terminals in $\mathcal M_x$ contains three vertices. We infer that the total number of vertices on paths in $\mathcal P$ is at least
\begin{align*}
|\mathcal P_{st}| \cdot (k+1) + |\mathcal P_x| \cdot 3 &= k(n-1)(k+1) + 3\binom{k}{2}\\
&= k\left(n(k-1) + 2(n-1)\right) + \binom{k}{2} = |V(H)| \enspace .
\end{align*}
We infer that every path in $\mathcal P_{st}$ consists of $k+1$ vertices, and every path in $\mathcal P_x$ consists of three vertices.
In particular, for $i = 1,\hdots,k$ and $v \in V^i \setminus \{u^i\}$, the path in $\mathcal P_{st}$ that connects $s_v^i$ and $t_v^i$ goes either through~$X_v^i$ or~$X_{u^i}^i$.
Consequently, for $i = 1,\hdots, k$ there exists a vertex $v^i \in V^i$ such that the vertices of~$W^i$ that do not lie on any path from $\mathcal P_{st}$ are exactly the vertices on the path $X_{v^i}^i$. 

We claim that $\{v^i~|~i =1,\hdots,k\}$ is a clique in $G$. To this end, consider a pair $1 \leq i < j \leq k$.
Since~$|\mathcal P_x| = \binom{k}{2}$, there exists a path in $\mathcal P_x$ that goes through~$p_{i,j}$. Moreover, this path has exactly three vertices. Since the only neighbours of $p_{i,j}$ that are not used
by paths from $\mathcal P_{st}$ are $x_{v^i,j}^i$ and $x_{v^j,i}^j$, we infer that $(x_{v^i,j}^i, x_{v^j,i}^j) \in \mathcal M$ and, consequently, $v^iv^j \in E(G)$. This concludes the proof of the correctness of the construction.

\medskip
\noindent
\textbf{Bounding the feedback vertex set number.}
We are left with a proof that $H$ has bounded feedback vertex set number.

To this end, first observe that $H-P$ contains $k$ connected components, being the gadgets~$W^i$. Second, observe that the deletion of the endpoints of the path~$X_{u^i}^i$ from the gadget $W^i$ breaks $W^i$ into connected components being paths on at most $k+1$ vertices.
Consequently,~$H$ has a feedback vertex set~$R$ consisting of~$P$ and $\{x_{u_i,1}^i,x_{u_i,k}^i\in V(W^i)~|~i=1,\hdots,k\}$, of size $|R| = \mathcal O(k^2)$.
This finishes the proof of Theorem~\ref{thm:ndp-w1hard-r}. \qed

\section{Hardness of Edge-Disjoint Paths in Almost-Forests}
\label{sec:hardnessofedpinalmostforests}
In this section we show that {\sc EDP} (and hence {\sc MaxEDP}) is $\mathsf{NP}$-hard already in graphs that are almost forests, namely, in graphs that are forests after deleting two nodes.
That is, we prove Theorem~\ref{thm:edp-nphard-r2}.

\begin{proof}[Proof of Theorem~\ref{thm:edp-nphard-r2}]
  We first show $\mathsf {NP}$-hardness of EDP for $r=2$.
  We reduce from the problem {\sc Edge 3-Coloring} in cubic graphs, which is $\mathsf{NP}$-hard~\cite{Holyer1981}.
  Given a cubic graph~$H$, we construct a complete bipartite graph $G$, where one of the two partite classes of $V(G)$ consists of three nodes $\{v_1,v_2,v_3\}$, and the other partite class consists of $V(H)$.
  As terminal pairs, we create the set~${\mathcal M = \{(s,t)~|~\{s,t\}\in E(H)\}}$; in words, we want to connect a pair of nodes by a path in $G$ if and only if they are connected by an edge in~$H$.
  This completes the construction of the instance $(G,\mathcal M)$ of {\sc MaxEDP}.
  Notice that $G$ has a feedback vertex set of size $r = 2$, since removing any size-2 subset of $\{v_1,v_2,v_3\}$ from $G$ yields a forest.
  
  Regarding correctness of the reduction, we show that $H$ is 3-edge-colorable if and only if \emph{all} pairs in~$\mathcal M$ can be routed in $G$.
  
  In the forward direction, suppose that $H$ is 3-edge-colorable.
  Let~${\varphi:E(H)\rightarrow\{1,2,3\}}$ be a proper 3-edge-coloring of $H$.
  For $c = 1,2,3$, let~${E_c\subseteq E(H)}$ be the set of edges that receive color~$c$ under $\varphi$.
  Then there is a routing in $G$ that routes all terminal pairs $\{(s,t)\in\mathcal M~|~\{s,t\}\in E_c\}$ exclusively via the node~$v_c$ (and thus via paths of length~2).
  Notice that this routing indeed yields edge-disjoint paths, for if there are distinct vertices~$s,t_1,t_2\in V(H)$ and edges $e_1 = \{s,t_1\},e_2 = \{s,t_2\}\in E(H)$, then $e_1,e_2$ receive distinct colors under~$\varphi$ (as $\varphi$ is proper), and so the two terminal pairs $\{s,t_1\},\{s,t_2\}$ are routed via distinct nodes~$c_1,c_2\in\{v_1,v_2,v_3\}$, and thus also via edge-disjoint paths.
  
  In the backward direction, suppose that all terminal pairs in $\mathcal M$ can be routed in~$G$.  Since~$H$ is cubic, any node $s\in V(H)$ is contained in three terminal pairs.  Therefore, no path of the routing can have a node in~$V(H)$ as an internal node and thus all paths in the routing have length~2.
  Then this routing naturally corresponds to a proper 3-edge-coloring $\varphi$ of~$H$, where any terminal pair $\{s,t\}$ routed via $c$ means that we color the edge $\{s,t\}\in E(H)$ with color ~$c$ under $\varphi$.

  \medskip
In order two show $\mathsf{NP}$-hardness of \textsc{MaxEDP} for $r=1$, we also reduce from \textsc{Edge 3-Coloring} in cubic graphs and perform a similar construction as described above: This time, we construct a bipartite graph $G$ with one subset of the partition being~$\{v_1,v_2\}$, the other being $V(H)$, and the set $\mathcal{M}$ of terminal pairs being again specified by the edges of $H$.  This completes the reduction. The resulting graph $G$ has a feedback vertex set of size $r=1$.

We claim that $H$ is 3-colorable if and only if we can route $n=|V(H)|$ pairs in $G$.  

In the forward direction, suppose that $H$ is 3-edge-colorable.  
Let~${\varphi:E(H)\rightarrow\{1,2,3\}}$ be a proper 3-edge-coloring of $H$.
For $c = 1,2,3$, let~${E_c\subseteq E(H)}$ be the set of edges that
receive color $c$ under $\varphi$.  Then there is a routing in $G$
that routes all f $\{(s,t)\in\mathcal M~|~\{s,t\}\in E_c\}$
exclusively via the node $v_c$ (and thus via paths of length~2) for
the colors $c=1,2$. (The terminals corresponding to edges receiving
color 3 remain unrouted.)

The reasoning that the resulting routing is feasible is analogous to the case of~$r=2$.  Since for each of the~$n$ terminals exactly two of the three terminal pairs are routed, this means that precisely $n$ terminal pairs are routed overall.

In the backward direction, suppose that $n$ terminal pairs in $\mathcal M$ can be routed in~$G$.  
Since any terminal~$v$ in $G$ is a node in $V(H)$ has therefore has degree two in $G$, this means that at most two paths can be routed for $v$.  As~$n$ terminal pairs are realized, this also means that \emph{exactly} two paths are routed for each terminal.  Hence, none of the paths in the routing has length more than two. Otherwise, it would contain an internal node in~$V(H)$, which then could not be part of two other paths in the routing. 
Then this routing naturally corresponds to a partial edge-coloring of~$H$, where any terminal pair $\{s,t\}$ routed via $c$ means that we color the edge $\{s,t\}\in E(H)$ with color $c$.  Since each terminal $v$ in $V(H)$ is involved in exactly two paths in the routing, exactly one terminal pair for $v$ remains unrouted. Hence, exactly one edge incident on~$v$ in $H$ remains uncolored in the partial coloring.  We color all uncolored edges in~$H$ by color 3 to obtain a proper~3-coloring.
\end{proof}

Thus, we almost close the complexity gap for {\sc EDP} with respect to the size of a minimum feedback vertex set, only leaving the complexity of the case $r = 1$ open.
We conjecture that this case can be solved in polynomial time.

\bibliography{routing-treewidth}
\bibliographystyle{abbrv}

\end{document}